\documentclass[12pt]{article}
 \usepackage{amsthm}
\usepackage{lineno}
\usepackage{graphicx}
\usepackage{multirow}
\usepackage[labelfont=bf]{caption}
\usepackage{subfigure}
\usepackage{booktabs}
\usepackage{amsfonts}
\usepackage[hidelinks]{hyperref}
\usepackage{setspace}
\newcommand{\ie}{{\it i.e.} }
\newcommand{\re}{\mathbb{R}}

\newcommand{\rep}{\mathbb{R}_+}

\newcommand{\Pp}{\mathcal{P}}
\newcommand{\X}{\mathcal{X}}

\newtheorem{theorem}{Theorem}
\newtheorem{lemma}[theorem]{Lemma}
\newtheorem{proposition}[theorem]{Proposition}

\begin{document}
\title{
A Note on the Identifiability of Generalized Linear Mixed Models
}
\author{Rodrigo Labouriau
	    \thanks{
	      Department of Molecular Biology and Genetics, 
	      Aarhus University}}
\date{May 2014}
\maketitle

\begin{abstract}
I present here a simple proof that, under general regularity conditions, the standard parametrization of generalized linear mixed model is identifiable. The proof is based on the assumptions of generalized linear mixed models on the first and second order moments and some general mild regularity conditions, and, therefore, is extensible to quasi-likelihood based generalized linear models. In particular, binomial and Poisson mixed models with dispersion parameter are identifiable when equipped with the standard parametrization.
\end{abstract}

%
%
\newpage

\section{Introduction}

Here I consider the problem of establishing the identifiability of the standard parametrisation of a generalized linear mixed mode (GLMM). 
The concept of identifiability that I use here is the following (see Bickel and Dokuson, 1977, page 60, Lehman 1983 and J{\o}rgesen and Labouriau 2012).
Consider a statistical model 
\begin{eqnarray} \label{Eq00}
 \Pp = \left \{ P_\theta : \theta \in \Theta \right \}
\end{eqnarray}
parametrised by $\theta $, \ie $\Pp$ is a family of probability distributions defined in the same measurable space indexed by the parameter $\theta$. 
The parametrisation used in (\ref{Eq00}) is {\it identifiable} when the mapping $\theta \mapsto P_\theta$ is one-to-one, \ie 
\begin{eqnarray} \nonumber
 \theta_1 \ne \theta_2 
 \Longrightarrow 
 P_{\theta_1} \ne  P_{\theta_2}  \, ,
\end{eqnarray}
for each $ \theta_1$ and $ \theta_2$ in $\Theta$.

Without loss of generality we consider a GLMM with one single random component, say $U$ and one fixed effect $x$. Denote by $Y$ the random variable representing the response for one observation. I assume that conditionally on $U$, $Y$ is distributed according to an exponential dispersion model and that 
\begin{eqnarray} \label{Eq001}
 E \left ( Y \vert U=u \right ) = h  \left ( x\beta + u  \right ) \,\, ,
\end{eqnarray}
where $\beta $ is a parameter and $h $ is the response function (\ie the inverse of the link function), which is assumed to be monotone (\ie increasing or decreasing) and smooth. The random component is assumed to be continuously distributed with expectation $0$ and variance $\sigma^2$ (typically assumed to be normally distributed) and have a probability density of the form $\phi (\, \cdot \, / \sigma)$.
Futhermore, I assume that 
\begin{eqnarray} \label{Eq002}
 Var \left ( Y \vert U=u \right ) = \xi \, V  \left [   h  \left ( x\beta + u  \right ) \right )] \,\, ,
\end{eqnarray}
where $\xi$ is the dispersion parameter and $V$ is the variance function, which associates the mean to the variance (see J{\o}rgensen {\it et al}, 1996, Breslow and Clayton, 1983).

The GLMM referred above is then parametrised by $\theta = (\beta, \sigma^2, \xi )$. I will show that under mild regularity conditions this parametrisation is identifiable. That is, denoting  the distribution of $Y$ when the parameter takes the value $\theta = (\beta, \sigma^2, \xi )$ by $P_\theta = P_{\beta, \sigma^2, \xi}$, we will show that, if 
$\theta_1 = (\beta_1, \sigma^2_1, \xi_1 ) \ne \theta_2 = (\beta_2, \sigma^2_2, \xi_2 ) $, then $P_{\theta_1} \ne P_{\theta_2}$. The proof will be established using two propositions. First I will show in proposition 1 that if $\beta_1 \ne \beta_2$ or $\sigma^2_1 \ne \sigma^2_2$ (no matter the values of $\xi_1$ and $\xi_2$) then $P_{\theta_1} \ne P_{\theta_2}$. This proof will use the property  (\ref{Eq001}) and the fact that if two distributions have different expectation, then they are not equal. Next I will show in proposition 2 that if $\xi_1 \ne \xi_2$, 
then $P_{\beta, \sigma^2, \xi_1} \ne P_{\beta, \sigma^2, \xi_2}$ for any values of $\beta$ and $\sigma$, which will complete the proof. I will call that the part 2 of the proof. The proof will use the relation (\ref{Eq002}) and the fact that if the variance of two distributions are different, then the distributions are different.

\section{Preparatory basic calculations}

Before embracing the proof, I calculate the expectation and  the variance of $Y$. These calculations will use implicitly the following two regularity conditions. For all $k\in\re$ and each $\sigma\in\rep$,
\begin{description}
\item {i)}
 The integral $\int \left [ h (  k + \sigma z )\right ]^2  \phi ( z ) dz$ is finite
 \item {ii)}
 The integral $\int V \left [ h (  k + \sigma z )\right ]  \phi ( z ) dz$ is finite.
\end{description}
Here $\phi $ is the density of the distribution of the random component (typically assumed to be normally distributed) and the integration is in the support of the distribution of the random component, which will be assumed to be $\re$.
Under condition i) the expectation of $Y$ is given by
\begin{eqnarray}  \nonumber 
E(Y) = E \left [ E(Y \vert U ) \right ] 
        & = & \int_\re h ( x\beta + u) \, \phi (u/\sigma ) \, du 
        \\ \nonumber 
        & & \mbox{(making a change of variable)}
        \\ \label{Eq0003}
        & = & 1/\sigma \int_\re h ( x\beta + \sigma z ) \, \phi (z ) \, dz 
        = \mu_{\beta, \sigma} .
\end{eqnarray}
On the other hand, the variance of $Y$ is given by
\begin{eqnarray} \label{eq01}
 Var (Y) = Var \left [ E(Y\vert U) \right ] + E  \left [ Var(Y\vert U) \right ]
 \, .
\end{eqnarray} 
Now, 
\begin{eqnarray} \nonumber
  Var\left [ E(Y\vert U) \right ] & = & Var  \left [ h(x\beta + U ) \right ] 
  \\ \nonumber & = &
  \int_\re h^2 (  x\beta + u )  \phi ( u/\sigma ) du 
  - \left [  \int_\re h( x\beta + u )  \phi ( u/\sigma ) du \right  ]^2
  \\ \label{eq02} & = &
  \zeta_{\beta, \sigma} - \mu_{\beta, \sigma}^2 
  \, ,
\end{eqnarray} 
where $\zeta_{\beta, \sigma} =  \int_\re h^2 (  x\beta + u )  \phi ( u/\sigma ) du $. Furthermore,
\begin{eqnarray} \label{eq03}
 E \left [ Var (Y\vert U) \right ] & = & 
 \xi \int_\re  \, V  \left [   h  \left ( x\beta + u  \right ) \right ]\phi ( u/\sigma ) du 
 = \xi \upsilon_{\beta , \sigma} \, ,
\end{eqnarray} 
where $ \upsilon_{\beta , \sigma}  = \int_\re  \, V  \left [   h  \left ( x\beta + u  \right ) \right ]\phi ( u/\sigma ) du $. Inserting (\ref{eq02}) and (\ref{eq03}) in (\ref{eq01}) yields
\begin{eqnarray} \label{eq04}
 Var (Y) =   \zeta_{\beta, \sigma} - \mu_{\beta, \sigma}^2 + \xi \upsilon_{\beta , \sigma} \, . 
\end{eqnarray} 

\section{Proof of the indentifiability}

\subsection{Identifiability of the fixed and random effect parameters}

\begin{proposition}
For each $\theta_1 = (\beta_1, \sigma^2_1, \xi_1)$ and $\theta_2 = (\beta_2, \sigma^2_2, \xi_2)$ such that $\beta_1 \ne \beta_2$ or $\sigma^2_1 \ne \sigma^2_2$, $P_{\theta_1} \ne P_{\theta_2}$, provided the following two conditions are fulfilled:
\begin{description}
 \item{1)} 
 The integral $\int  h (  k + \sigma z )  \phi ( z ) dz$ is finite
 \item{2)} 
 The explanatory variable $x$ takes values in the set $\X$ and the equation
\begin{eqnarray} \label{Eqq09}
  \frac{h(x\beta_1+\sigma_1z)}{\sigma_1} = 
  \frac{h(x\beta_2+\sigma_2z)}{\sigma_2}
  \,\,
  \mbox{ for all } x \mbox{ in } \X \mbox{ and all }
  z \mbox{ in } \re
\end{eqnarray}
has no solution.
\end{description}
\end{proposition}
\noindent
Note that in the lemma above nothing is said about the values of $\xi_1$ and $\xi_2$.
\begin{proof}
Suppose, by hypothesis of absurdum, that there exist $\beta_1, \beta_2$, and $\sigma^2_1, \sigma^2_2$ such that $P_{\beta_1, \sigma^2_1, \xi_1} = P_{\beta_2, \sigma^2_2, \xi_2}$. In particular, we have that the expectation of $Y$ is equal for the distributions indexed by the two parameters, so that $\mu_{\beta_1, \sigma_1}  = \mu_{\beta_2, \sigma_2}$. That is, 
\begin{eqnarray} \nonumber
   \int_\re h ( x\beta_1 + \sigma_1 )/\sigma_1 \phi (z ) dz  = 
   \int_\re h ( x\beta_1 + \sigma_2 ) /\sigma_2 \phi (z ) dz  \, ,
   \mbox{ for all } x \mbox{ in } \X \, .
\end{eqnarray}  
Moreover, the equality above holds also for any interval $I$ in $\re$
\begin{eqnarray} \label{EqQQQ}
   \int_I h ( x\beta_1 + \sigma_1 )/\sigma_1 \phi (z ) dz  = 
   \int_I h ( x\beta_1 + \sigma_2 ) /\sigma_2 \phi (z ) dz  \, ,
   \mbox{ for all } x \mbox{ in } \X \, ,
\end{eqnarray} 
which corresponds to say that the two distributions have the same conditional expectations given the result are in the interval $I$.
Since the interval $I$ is arbitrary and $h$ and $\phi$ are continuous,
the technical lemma \ref{TLemma} implies that
\begin{eqnarray} \nonumber
  \frac{h(x\beta_1+\sigma_1z)}{\sigma_1} = 
  \frac{h(x\beta_2+\sigma_2z)}{\sigma_2}
  \,\,
  \mbox{ for all } x \mbox{ in } \X \mbox{ and all }
  z \mbox{ in } \re \, ,
\end{eqnarray} 
which contradicts the hypothesis 2).
\end{proof}

\subsection{Identifiability of the dispersion parameter}
\begin{proposition}
Under the regularity conditions i) and ii), 
$P_{\theta_1} \ne P_{\theta_2}$, for each $\theta_1 = (\beta, \sigma^2, \xi_1)$ and $\theta_2 = (\beta, \sigma^2, \xi_2)$ such that $\xi_1 \ne \xi_2$. 
\end{proposition}
\begin{proof}
Suppose, by hypothesis of absurdum, that there exist $\beta, \sigma^2, \xi_1$ and $\xi_2$ with $\xi_1\ne \xi_2$ such that $P_{\beta, \sigma^2, \xi_1} = P_{\beta, \sigma^2, \xi_2}$. Then, $Var(Y)$ under $P_{\beta, \sigma^2, \xi_1}$ must be equal to $Var(Y)$ under $P_{\beta, \sigma^2, \xi_2}$, which implies by (\ref{eq04}) that
\begin{eqnarray} \nonumber
  \zeta_{\beta, \sigma} - \mu_{\beta, \sigma}^2 + \xi_1 \upsilon_{\beta , \sigma}
  = 
   \zeta_{\beta, \sigma} - \mu_{\beta, \sigma}^2 + \xi_2 \upsilon_{\beta , \sigma} \, ,
\end{eqnarray} 
wich is equivalent to
\begin{eqnarray} \nonumber
   \xi_1 \upsilon_{\beta , \sigma}
  = 
   \xi_2 \upsilon_{\beta , \sigma} \, ,
\end{eqnarray} 
which implies that $\xi_1 = \xi_2$ since $ \upsilon_{\beta , \sigma} \ne 0$. 
But this contradicts the hypothesis of absurdum and therefore the proof is concluded.
\end{proof}

\section{Closing remarks} 

Although the proof above is general there are two particular examples that are of interest in many practical situations: the binomial and the Poisson GLMM. In those cases the dispersion parameter $\xi$ will represent the parameter used for modelling under- and over-dispersion via quasi-likelihood. When the dispersion parameter $\xi$ is equal to one the GLMM corresponds to the classical binomial and Poisson standard models; and in this case, the identifiability follows from the proposition 1. 

For the sake of illustration here is the verification of the identifiability of the standard parametrisation of a Poisson model with the logarithmic link, \ie the response function is $h(\, \cdot \,  ) = \exp (\, \cdot \, )$. Equation (\ref{Eqq09}) in this particular example is 
\begin{eqnarray} \nonumber
 \frac{\exp(x\beta_1 + \sigma_1 z)}{\sigma_1} =
 \frac{\exp(x\beta_2 + \sigma_2 z)}{\sigma_2}
 \,\,
  \mbox{, for all } x \mbox{ in } \X \mbox{ and all }
  z \mbox{ in } \re \,\, ,
\end{eqnarray}
which is equivalent to (taking logarithms in both sides) 
\begin{eqnarray} \nonumber
 x\beta_1 + \sigma_1 (z-1) =
  x\beta_2 + \sigma_2 (z-1)
 \,\,
  \mbox{, for all } x \mbox{ in } \X \mbox{ and all }
  z \mbox{ in } \re \,\, .
\end{eqnarray} 
But the equation above has no solution if $\X$ has more than one element. Analogous arguments yields the identifiability of a binomial model with the classic logistic link or the probit link.

The techniques for establishing identifiability used here can be easily applied, after a straightforward adaptation, to a context of survival analysis for the piecewise constant hazard model or the discrete time proportional hazard model, both with frailties and dispersion parameter (see Maia {\it et al}, 2014 for details on those models).

\section*{Acknowledgements}

This work was partially financed financed by the  {\it Funda\c c\~ao Apolodoro Plaus\^onio}.

\newpage

\appendix

\section{A trivial technical lemma}

I prove below a trivial technical lemma from basic calculus, which must be generally known, but I include here for completeness of this note.
\begin{lemma}\label{TLemma}
 If $f$ and $g$ are two continuous and integrable real functions such that for each interval $I$ in $\re$,
 \begin{eqnarray}\label{EqLLL}
  \int_I f(x) dx =  \int_I g(x) dx \, ,
\end{eqnarray}
then $f(x) = g(x)$ for all $x\in \re$.
\end{lemma}
\begin{proof}
Suppose, by hypothesis of absurdum, that there is a $y\in\re$ such that $f(y)\ne g(y)$, say $f(y)<g(y)$. Then, from the continuity of $f$ and $g$, there exist an interval $I$ such that $f(x)<g(x)$ for all $x\in I$. But then $\int_I f(x) dx < \int_I g(x) dx$, which contradicts (\ref{EqLLL}).
\end{proof}

\end{document}